\documentclass[12pt]{article}
\usepackage{amsmath,amsfonts, amsthm,amssymb,hyperref,bbm,eucal,verbatim,enumerate,soul}
\usepackage{color} 
\newtheorem{thm}{Theorem}
\newtheorem{lemma}{Lemma}[section]

\newtheorem{rmk}[lemma]{Remark}
\newtheorem{cor}[lemma]{Corollary}

\newtheorem*{bel*}{Belief}
  
\numberwithin{equation}{section}

\newtheoremstyle{named}{}{}{\itshape}{}{\bfseries}{.}{.5em}{\thmnote{#3}}
\theoremstyle{named}
\newtheorem*{thm1*}{Theorem}

\def\K{\widetilde{K}}
\def\E{\widetilde{E}}
\def\I{\widetilde{I}}
\def\s{\widetilde{s}}
\def\g{\widetilde{g}}
\def\hr{{{{\it hr}}}}

\def\ind{\mathbbm{1}}

\newcounter{remnr1}
\def\res1{
    \addtocounter{remnr1}{1}
    \vspace{2mm}\noindent{\bf (\Alph{remnr1})} }

\begin{document}

\title{Semi-classical analysis of non self-adjoint
transfer matrices in statistical mechanics. I}
\author{Margherita Disertori\textsuperscript{1} and Sasha Sodin\textsuperscript{2}}
\footnotetext[1]{Institute for Applied Mathematics \& Hausdorff Center for Mathematics,
University of Bonn, Endenicher Allee 60, 53115 Bonn, Germany. \\
E-mail: margherita.disertori@iam.uni-bonn.de}
\footnotetext[2]{Department of Mathematics, Princeton University,
Fine Hall, Washington Road, Princeton, NJ 08544-1000, USA \& 
School of Mathematical Sciences, Tel Aviv University,
Ramat Aviv,
Tel Aviv 699780,
Israel. \\
E-mail: sashas1@post.tau.ac.il}
\maketitle

\begin{abstract}
We propose a way to study one-dimensional statistical mechanics models 
with com\-plex-valued action
using transfer operators. The argument consists of two steps. First, the contour of integration is 
deformed so that the associated transfer operator is a perturbation of
a normal one. Then the transfer operator is studied using methods of semi-classical
analysis. 

In this paper we concentrate on the second step,
the main technical result being
a semi-classical estimate for powers of an integral operator which is approximately normal.
\end{abstract}

\section{Introduction}
Operator-theoretic methods are known to be of great help in one-dimensional
statistical mechanics. 

\vspace{1mm}\noindent
 Consider the following prototypical example. Let $V:\mathbb{R} \to \mathbb{R}$ be a potential growing 
sufficiently fast at infinity.
It is known that, for any value of $W > 0$ , 
there exists a unique probability measure 
(Gibbs measure) $\mu_{V, W}$ on  the space of configurations in one dimension $\mathbb{R}^\mathbb{Z}$ 
such that, for every $M,N \in \mathbb{N}$, the conditional probability density 
at $\phi \in \mathbb{R}^\mathbb{Z}$ given $\phi|_{\{-M,\cdots,N\}^c}$
is proportional to
\begin{equation}\label{eq:model}
\exp \left\{ - \sum_{j =-M}^N V(\phi_j) - \sum_{j=-M-1}^{N} W^2 (\phi_j - \phi_{j+1})^2 \right\}~. 
\end{equation}
The existence of $\mu_{V,W}$ is a consequence of general theory, independent of the
dimension of the lattice  (see \cite[Chapter 7]{Ru}).
The uniqueness can be proved using the transfer
matrix formalism described below; it also follows for example from the van Hove theorem as stated in
the book of Ruelle \cite[Section 5.6.6]{Ru},
combined with the Dobrushin--Shlosman theorem \cite{DS}.

One says that the measure $\mu_{V, W}$ corresponds to the (real-valued) action
\begin{equation}\label{eq:Ham}
S(\phi) = \sum_j V(\phi_j) + \sum_j W^2 (\phi_j - \phi_{j+1})^2~.
\end{equation}

The properties of $\mu_{V, W}$, such as exponential decay of correlations 
(between $\phi_j$ and $\phi_k$ as $|j-k|\to \infty$), are encoded
in the spectral structure of the self-adjoint operator $K$ (called the transfer
operator, or transfer matrix), acting on $L_2(\mathbb{R})$ as an integral operator with kernel given by 
\begin{equation}\label{eq:saK}
 K(x, y) = \exp \left\{ - W^2(x-y)^2 - \frac{V(x) + V(y)}{2}\right\}~.
\end{equation}
When $W \gg 1$ is large,
semi-classical analysis allows to relate the spectral properties of $K$
to those of a simpler operator $\K$,  which depends
only on the behaviour of $V$ at its minima. If the minima of $V$ are non-degenerate, the potential
corresponding to $\K$ is quadratic, and thus $\K$ is referred
to as the harmonic approximation to $K$. In this case, computations can be performed explicitly; the small
quantity $W^{-1}$ plays the r\^ole of the semi-classical parameter $\hbar$.

In the case when $V$ has a unique non-degenerate minimum $V(0)=0$, $\K$ is given by the harmonic oscillator
\begin{equation}\label{eq:saKtil}
\K(x, y) = \exp \left\{ - W^2(x-y)^2 - \frac{V''(0)}{4} (x^2 + y^2)\right\}~.
\end{equation}
When $V$ has several minima, $K$ is approximated by a direct sum of several
harmonic oscillators.

The semi-classical approach to various problems of one-dimensional statistical
mechanics is presented in detail in the monograph of Helffer \cite{H}.

\vspace{2mm}\noindent
On the other hand, in many questions in statistical mechanics the potential $V$ is complex-valued. This problem, often referred to as the ``sign problem'' or 
``complex action problem'', is inherent, for example, to lattice quantum chromodynamics 
(see for example Muroya et al.\ \cite{Muroya} and Splittorff and Verbaarschot \cite{SV}), and
also arises in supersymmetric models appearing in the study of random operators
(see the reviews of Spencer \cite{Sp,Sp2}).

A na\"ive attempt to apply the methods taylored for real-valued action
to this situation encouters immediate obstacles. In the context of transfer
operators, neither $K$ nor $\K$ is self-adjoint;
thus perturbation theory is not easily set on a rigorous basis, and on the 
other hand the spectrum of $K$
is not directly connected to the semigroup $(K^n)_{n \geq 0}$. We refer to the articles
of Davies \cite{D1,D2}  and further to the monograph of 
Helffer \cite[Chapters 13--15]{H2} and to the PhD thesis of Henry \cite{RH},
where  difficulties in the semi-classical analysis of 
non self-adjoint operators are discussed, along with some positive results.

Our goal in this paper is to suggest a strategy which allows to apply semi-classical
analysis to models with complex-valued action, in spite of these difficulties. Here we apply it to a toy model (with one
saddle); in a subsequent work, we hope to apply it to a statistical mechanics model
arising from the supersymmetric
analysis of a class of random band matrices; see \cite{DPS} for an analysis of
a related three-dimensional model, and the review of Spencer \cite{Sp} for a  discussion of 
supersymmetric
models arising from random band matrices, and the possible transfer matrix approach.

\vspace{2mm}\noindent
The strategy we suggest is as follows. Before setting up the transfer operator,
we deform the contour of integration so that the harmonic approximation $\K$
is almost normal (in appropriate sense). 
Then we set up the transfer operator and analyse it 
using (semi-) classical tools.  
In this paper, we restrict ourself to the simplest
deformation
\begin{equation}\label{eq:deform}
\phi \gets \zeta \phi, \quad |\zeta| = 1~;
\end{equation}
in general, a more complicated deformation (similar to (\ref{eq:deform}) near each
saddle point but different away from the saddles) may be required.

To motivate the idea, let us consider the differential operator
\[
L = - \frac{1}{W^2} \frac{d^2}{dx^2}  + (a+ib) x^2~, \quad a > 0.
\]
One can always find $\zeta$, $|\zeta|=1$, so that after the change of variables
$x \gets \zeta x$ the operator $L$ becomes normal, i.e.\ a scalar multiple of
\[ 
\hat{L} = -\frac{1}{W^2} \frac{d^2}{dx^2} + |a+ib| x^2~. 
\]
In a similar way, for an  integral operator $K$, one can  rotate the 
contour so that the harmonic approximation to $K$ near the saddle point  (i.e.\ an 
operator with quadratic $V$) becomes normal.

The main result of this paper justifies the approximation $K\approx \K$
for a class of operators $K$  of a form similar to (\ref{eq:saK})  by the corresponding
harmonic approximation $\tilde{K}$, in the case
when $\tilde{K}$ is  almost normal. The precise statement and conditions
are given in Section \ref{main} below.
The proof of this result occupies the central
part of this paper, and appears in Section~\ref{s:main}. It is preceded by Section~\ref{s:harm}, in which
several properties of the non-self-adjoint
harmonic oscillator are collected.

In Section~\ref{s:applic} we show an application
to a statistical mechanics model with complex-valued potential;
namely, we find the sharp exponential decay of correlations for this model.

\section{Main result}
\label{main}

\subsection{Statement of the main technical result}\label{sub:assump}

Let $K: L_2(\mathbb{R}) \to L_2(\mathbb{R})$ be an operator defined by
its kernel
\begin{equation}\label{eq:defK}
K(x, y) = \exp \left\{ - W^2 \zeta^2 (x-y)^2 - \frac{1}{2} U(x) - \frac{1}{2}U(y) \right\}~,
\end{equation}
where $\zeta \in \mathbb{C}$ is a complex number with $|\zeta| = 1$ and $\Re \zeta^2 > 0$;
$W>0$ is a large parameter, and the potential $U: \mathbb{R} \to \mathbb{C}$ satisfies the following assumptions:
\begin{enumerate}[U1)]
\item $U$ is smooth and  $U(0) = U'(0) = 0$; 
\item $\zeta^2 U''(0) > 0$; [this condition ensures that the operator is approximately normal near the saddle]
\item $\Re U(x) \geq \frac{1}{C} \min(1, |x|^2)$ for all $x \in \mathbb{R}$; [in particular, $0$ is the unique
minimum of $\Re  U$]
\item  $U$ has an analytic extension to a strip $|\Im z| \leq c$ about the real axis
which    satisfies 
$|U'(z)| \leq C  (\max (1, \Re [U (z)]))^\gamma$ 
for some $\gamma > 1$ and all 
$z$ in this strip. 
      \newcounter{enumU_saved}
      \setcounter{enumU_saved}{\value{enumi}}
\end{enumerate}
We use the analyticity assumption to justify saddle point approximations; the second
red part of the condition is a mild regularity assumption which rules out wildly oscillating potentials such as
$x^2(1 + \exp(\sin e^x))$.
\begin{thm1*}[Main Proposition]
Let $K$ be an operator given by (\ref{eq:defK}), where $U$ satisfies the assumptions
U1)--U4). Denote
\begin{equation}\label{eq:defAlpha}
\alpha = W \sqrt{\zeta^2 U''(0)/2}~, \quad
\mu = \sqrt{\frac{\pi}{W^2\zeta^2 + \alpha}}~,
\end{equation}
and
\begin{equation}\label{eq:defgalpha}
g_\alpha(x) = \left(\frac{2\alpha}{\pi}\right)^{1/4} \exp(-\alpha x^2)~.
\end{equation}
Let
\begin{equation}\label{eq:matrix} K = \sqrt{\frac{\pi}{W^2 \zeta^2 + \alpha}}
\left(  \begin{array}{cc} A & B \\ C & D \end{array} \right) =
\mu \left(  \begin{array}{cc} A & B \\ C & D \end{array} \right)
\end{equation}
be the block representation of $K$  corresponding to the decomposition
$L_2(\mathbb{R}) = \mathbb{C}g_\alpha \oplus (\mathbb{C}g_\alpha)^\perp$; more formally, if $\hat{K} = \mu^{-1}K$, and
$P$ is the orthogonal projection to $g_\alpha^\perp$,
\[  \left(  \begin{array}{cc} A & B \\ C & D \end{array} \right) =
\left(  \begin{array}{cc} (1-P) \hat{K} (1-P) & (1-P) \hat{K} P \\ P \hat{K} (1-P) & P \hat{K} P \end{array} \right)~. \]
Then
\[
\left(  \begin{array}{cc} A & B \\ C & D \end{array} \right)
= \left(  \begin{array}{cc}1 + O(W^{-1-\delta}) & O(W^{-1-\delta}) \\
O(W^{-1-\delta}) & \text{of norm $\leq 1 - \sqrt{\frac{|U''(0)|}{2}} \frac{\Re \zeta^2}{W}
+ O(W^{-1-\delta})$} \end{array}\right)~,
\]
meaning that
\begin{align} A &= 1 + O(W^{-1-\delta})~, \label{Abound} \\
\|B\|, \|C\| &= O(W^{-1-\delta})~, \\
\|D\| &\leq 1 - \sqrt{\frac{|U''(0)|}{2}} \frac{\Re \zeta^2}{W}
+ O(W^{-1-\delta})~; \label{Dbound}
\end{align}
the exponent $\delta > 0$ depends only on $\gamma$, and the implicit constants in the $O$-notation
depend on $\gamma$ and the implicit  constants in the assumptions.
\end{thm1*}
Note that by assumption U2) above $\alpha>0$, hence $g_{\alpha }$ is a real valued function.
Here and forth, we slightly abuse notation and identify scalar multiples of $1-P$ with complex numbers.

\subsection{The main corollary}
Let $\lambda_{0}$ be the largest eigenvalue (in absolute value) of $K$; the existence of $\lambda_0$
is part of the statement of Corollary~\ref{cor} below. Let $u_{0}$ be the 
corresponding eigenfunction. Since $u_{0}$ is complex valued, we can fix the following
normalisation conditions:
\begin{equation}\label{normcond}
\|u_{0} \|^{2}= \langle u_{0},u_{0} \rangle =1,\qquad  \mbox{and} \qquad 
\langle u_{0}, g_{\alpha } \rangle \geq 0,
\end{equation}
where $g_{\alpha }$ was defined in \eqref{eq:defgalpha} above.
The Main Proposition yields the following corollary:

\begin{cor}\label{cor} In the setting of the Main Proposition,  $K$ has a largest eigenvalue (in absolute
value), which satisfies
\begin{equation}\label{eq:lam0}
\lambda_0 = \sqrt\frac{\pi}{W^2\zeta^2 + \alpha} (1 + O(W^{-1-\delta})) = \mu (1 + O(W^{-1-\delta}))~,
\end{equation}
and the corresponding eigenfunction  $u_0$ (with the  normalisation conditions (\ref{normcond}))
satisfies
$\| u_0 - g_\alpha \| \leq CW^{-\delta}$.
For any natural $n$ and any $u$ in the invariant subspace $\bar{u}_0^\perp$ of $K$,
\begin{equation}\label{eq:semigroup} \|K^n u \| \leq |\lambda_0|^n
\left(  1 - \sqrt{\frac{|U''(0)|}{2}} \frac{\Re \zeta^2}{W} + O(W^{-1-\delta}) \right)^n \, \|u\|~. \end{equation}
\end{cor}

\noindent We remark that (\ref{eq:lam0}) can be restated as
\[
\lambda_{0} = \widetilde{\lambda}_{0} (1 + O(W^{-1-\delta}))~,
\]
where $\widetilde\lambda_0$ is the largest eigenvalue of the harmonic approximation
\begin{equation}\label{Ktilde}
\widetilde{K}(x, y) = \exp \left\{ -W^2 \zeta^2 (x-y)^2 - \frac{U''(0)}{4} (x^2 + y^2)\right\}~.
\end{equation}
This
remark is justified by the formul{\ae} of Section~\ref{s:harm} below.

\begin{proof}[Proof of Corollary~\ref{cor}]
According to the Main Proposition, $\hat{K} = \mu^{-1}K$ (with the normalising factor 
$\mu$ given by (\ref{eq:defAlpha}))
has the block structure
\[ 
\hat{K} = \left( \begin{array}{cc}
1 + O(W^{-1-\delta}) & O(W^{-1-\delta}) \\
O(W^{-1-\delta}) & D\end{array}\right)
\]
with respect to the decomposition
$L_2(\mathbb{R}) = \mathbb{C}g_\alpha \oplus (\mathbb{C}g_\alpha)^\perp$. Set
\[ 
\hat{K}_t = \left( \begin{array}{cc}
A_t  & B_t \\
C_t & D \end{array}\right) = t \hat{K} + (1-t)
\left( \begin{array}{cc}
1  & 0 \\
0 & D \end{array}\right)~, \quad 0 \leq t \leq 1~. 
\]
From \eqref{Dbound}, the largest eigenvalue of $\hat{K}_0$ is equal to $1$, 
whereas the rest of the spectrum lies in a disc of radius $1 - \mathrm{const}/W$. 
Let us show that, as $t$ varies from $0$ to $1$, the resolvent $R_z[\hat{K}_t] = (\hat{K}_t - z)^{-1}$ 
of $\hat{K}_t$ remains bounded on a circle $\mathcal{C}$ of radius $O(W^{-1-\delta})$ about $1$. 
This will imply that the spectral projection
\[ P_t = \frac{1}{2\pi i} \oint_\mathcal{C}  R_z[\hat{K}_t] dz \]
is a continuous function of $t \in [0,1]$, whence the rank of $P_t$ is  identically equal to $1$, 
its value at $t = 0$.

Let $R > 0$ be a large positive number, to be chosen shortly. We shall verify that the operator 
$D-z$ and the Schur complement 
\[ S_t(z) = (A_t - z) - B_t (D -z)^{-1} C_t\]
are invertible on the circle 
$|z-1| = RW^{-1-\delta}$  for all $t\in [0,1]$ (actually, the norm of
the inverse is bounded uniformly in $t$). Therefore we can apply
the Schur--Banachiewicz formula for the inverse of a block operator, according to which
the resolvent $R_z[\hat{K}_t]$ is given by
\[ R_z[\hat{K}_t] = \left(\begin{array}{cc}
S_t^{-1} & - S_t^{-1} B_t (D - z)^{-1} \\
- (D-z)^{-1} C_t S_t^{-1} & (D-z)^{-1} + (D-z)^{-1} C_t S_t^{-1} B_t (D-z)^{-1}
\end{array}\right)~.\]
One can choose $R>0$  such that
for all $z$ on the circle $|z-1| = RW^{-1-\delta}$ and all $t \in [0,1]$ one has 
$|A_{t}-z|\geq\frac{R}{2}W^{-1-\delta}$.
By  \eqref{Abound} this choice is independent of $W$. Hence the Schur complement
admits the bound 
\[ 
|S_t(z)| \geq \frac{R}{2}W^{-1-\delta} - O(W^{-2-2\delta} \times W)\geq  \frac{R}{4} W^{-1-\delta}~.
\]
Then we have:
\[ 
|S_t(z)^{-1}| \leq \frac{4}{R} W^{1+\delta}~. 
\]
Similarly, the norms of the other three blocks of $R_z[\hat{K}_t]$ are bounded, thus the application
of the Schur--Banachiewicz formula is justified, and
the norm of the resolvent is bounded for these $z$, uniformly in $0 \leq t \leq 1$. Hence
$\hat{K}_0 $ and $\hat{K}_1$ have the same number of eigenvalues in the interior of the 
circle $|z-1| \leq R W^{-1-\delta}$, i.e.\ $K$ has a unique eigenvalue $\lambda_0$ satisfying 
(\ref{eq:lam0}). The eigenvalue $\lambda_0$ is the largest one in absolute value due to the
estimate (\ref{eq:semigroup}) which we shall prove shortly.

\vspace{2mm}\noindent
Next, let $u_0$ be an  eigenfunction corresponding to $\lambda_0$, with normalisation
conditions \eqref{normcond}. Then we can
decompose $u_0$ as
\[ 
u_0 = p g_\alpha + u_0^\perp~, 
\quad  \text{where} \quad
p = \langle u_0, g_\alpha \rangle~, \quad
\langle u_0^\perp, g_\alpha \rangle = 0~; 
\]
the Main Proposition implies
\[ 
u_1 = \hat{K} u_0 = \left [(1+O(W^{-1-\delta}))p + O(W^{-1-\delta})\sqrt{1-|p|^2}
\right ] g_\alpha
+ u_1^\perp~, 
\]
where
\[ 
\| u_1^\perp\| \leq O(W^{-1-\delta}) |p| + \left (1 - \frac{1}{CW}\right ) \sqrt{1-|p|^2}~.
\]
On the other hand,
\[ u_1 = \hat{\lambda}_0 u_0 = p \hat{\lambda}_0 g_\alpha + \hat{\lambda}_0 u_0^\perp~, \]
where $\hat{\lambda}_0 = \mu^{-1} \lambda_0$;
comparing the norms of $u_1^\perp$ and $\hat{\lambda}_0 u_0^\perp$, using
(\ref{eq:lam0}) and $p=\langle u_{0}, g_\alpha \rangle\geq 0$, we obtain:
\[ 
\sqrt{1-|p|^2} \leq C W^{-\delta}~, \quad \text{whence} \quad
\| u_0 - g_\alpha\|  \leq  \| u_0^\perp\|+ \|  g_\alpha (p-1)\| \leq C' W^{-\delta}~. 
\]

\vspace{2mm}\noindent
Finally, let $u$ be a unit vector lying in the invariant subspace $\bar{u}_0^\perp$, of $\hat{K}$.
Then
\[ |\langle u, g_\alpha \rangle| \leq |\langle u, \bar{u}_0\rangle| +
|\langle u, \bar{u}_0 - g_\alpha \rangle| \leq CW^{-\delta}~. \]
Therefore
\[ u =  p' g_\alpha + u^\perp,~\quad u^\perp \perp g_\alpha~, \quad |p'|\leq CW^{-\delta}~.\]
Denote $q' = \|u^\perp\|=\sqrt{1-|p'|^{2}}$. Then
\begin{multline*}\| \hat{K} u \|^2
\leq \left\{ (1+O(W^{-1-\delta}))|p'| + O(W^{-1-\delta})|q'|\right\}^2 \\
+ \left\{ O(W^{-1-\delta})|p'|+\left  (1-\frac{c_0}{W} + O(W^{-1-\delta})\right)|q'|\right\}^2~,
\end{multline*}
where $c_0 = \sqrt{\frac{|U''(0)|}{2}} \Re \zeta^2$. Therefore
\[
\begin{split}
\|\hat{K} u \|^2
&\leq  \left [1 + O(W^{-1-\delta})\right ] |p'|^2 + \left [ (1-\frac{c_0}{W})^2 + O(W^{-1-\delta})
\right ] |q'|^2
+ O(W^{-1-\delta})\\
&\leq (1 - \frac{c_0}{W})^2 (|p'|^2 + |q'|^2) + O(W^{-1-\delta}) \\
&\leq \left (1 - \frac{c_0}{W} + O(W^{-1-\delta})\right )^2.
\end{split}
\]
Recalling that
\[ \langle \hat{K} u, \bar{u}_0 \rangle = \langle u, \hat{K}^* \bar{u}   _0 \rangle
= \langle u, \bar{\lambda}_0 \bar{u}_0 \rangle = 0~,  \]
we can iterate this estimate, thus obtaining
\[ 
\|\hat{K}^n u \| \leq \left (1 - \frac{c_0}{W} + O(W^{-1-\delta}) \right)^n
\]
for any $n$, as claimed.
\end{proof}

\section{Preliminaries: harmonic oscillator}\label{s:harm}

In this section, we collect the properties of the harmonic oscillator, defined by
\begin{equation}\label{eq:defHarm}
K_\hr(x, y) = \exp \left\{ - W^2 \zeta^2 (x-y)^2 - \frac{a+ib}{2} (x^2 + y^2)  \right\}
\end{equation}
for 
\begin{equation}\label{eq:welldef}
a > 0~, \quad b \in \mathbb{R}~, \quad |\zeta|=1~, \quad \Re \zeta^2 > 0~.
\end{equation}
The operator $K_\hr$ defined by (\ref{eq:defHarm}) is compact under the conditions (\ref{eq:welldef}), 
hence it has pure point spectrum. 

We are especially interested in the case when 
\begin{equation}\label{eq:welldef1}
\zeta^2 (a+ib)\in \mathbb{R}_{+}~,
\end{equation}
since this is consistent with assumption U2) for the operator $K$ defined in \eqref{eq:defK};
the properties stated here hold however in the generality of (\ref{eq:welldef}).

The eigenvalues of $K_\hr$ are given by the formula:
\begin{equation}\label{eq:evharm}
\lambda_j^{\hr} = \sqrt{\frac{\pi}{W^2 \zeta^2 + \alpha_\hr + \frac{a+ib}{2}}}
\left(\frac{W^2\zeta^2}{W^2\zeta^2+\alpha_\hr+\frac{a+ib}{2}} \right)^j~,
\end{equation}
where $\alpha_\hr$ is the solution of
\begin{equation}\label{eq:redefalphahr}
\alpha_\hr^2  = W^2 \zeta^2 (a+ib) + \frac{(a+ib)^2}{4}=  \alpha^{2}+ \frac{(a+ib)^2}{4}=
\alpha^{2}\left[1+ O (W^{-2}) \right]  ~,
\end{equation}
with $\Re \alpha_\hr > 0$, and $\alpha^{2}= W^2 \zeta^2 (a+ib)$.
This definition  is consistent with \eqref{eq:defAlpha}. 
Note that a solution with $\Re \alpha_\hr > 0$ exists as long as the right-hand side of 
(\ref{eq:redefalphahr}) is not real negative (this is ensured by (\ref{eq:welldef})). 
In the special  case when   both  \eqref{eq:welldef}
and  \eqref{eq:welldef1} hold we have
\[ 
\zeta^2 (a+ib)= \Re [\zeta^2 (a+ib)]= a/\Re \zeta^{2}>0~,
\]
hence for large $W$ the right-hand side of (\ref{eq:redefalphahr}) has positive real part.

The eigenfunction corresponding to $\lambda_0^{\hr} $ is exactly the function $g_{\alpha_\hr}$
given by (\ref{eq:defgalpha}) (with $\alpha$ replaced by $\alpha_\hr$); if $\alpha_\hr$ is real, the $L_2$ norm of 
$g_{\alpha_\hr}$ is equal to one.

\medskip
Let us comment on the validity of (\ref{eq:evharm})
for complex parameters. For any $z \in \mathbb{C}$, the Fredholm determinant
of the operator $z K_\hr$ is equal to
\begin{equation}\label{eq:Fr} \det(1 - z K_\hr) = \prod_{j=0}^\infty (1 - z \lambda_j^\hr)~.
\end{equation}
Also, $K_\hr$ is an analytic function of the parameters
$\zeta$ and $a+ib$, therefore, by Vitali's theorem, this remains true also 
for the left-hand side of (\ref{eq:Fr}), since it is the limit of a sequence of analytic functions
converging locally uniformly with respect to $\zeta$ and $a+ib$.
On the other hand, for real $\zeta$ and $a+ib$, the identity  (\ref{eq:evharm}) is well-known (see \cite[Section~5.2]{H}); thus the right-hand side of (\ref{eq:Fr}) is equal to 
\[ \prod_{j=0}^\infty \left(1 - z \sqrt{\frac{\pi}{W^2 \zeta^2 + \alpha_\hr + \frac{a+ib}{2}}}
\left(\frac{W^2\zeta^2}{W^2\zeta^2+\alpha_\hr+\frac{a+ib}{2}} \right)^j\right)~, \]
which is also an analytic function of $\zeta$ and $a+ib$. By analytic continuation, we have:
\[ \prod_{j=0}^\infty (1 - z \lambda_j^\hr)
= \prod_{j=0}^\infty \left(1 - z \sqrt{\frac{\pi}{W^2 \zeta^2 + \alpha_\hr + \frac{a+ib}{2}}}
\left(\frac{W^2\zeta^2}{W^2\zeta^2+\alpha_\hr+\frac{a+ib}{2}} \right)^j\right)\]
in the full range of parameters (\ref{eq:welldef}),
and this implies (\ref{eq:evharm}).

\vspace{3mm}\noindent
Now we turn to $K_\hr^* K_\hr$. Set $A = 2 W^2 \Re \zeta^2 + a$. One may check that
\begin{equation}\begin{split}\label{eq:K*K}
&(K_\hr^*K_\hr)(x, y)\\
&=
 \sqrt{\tfrac{\pi}{A}} \exp\Big\{ -
\frac{W^4 + W^2(a\bar\zeta^2 + (a-ib)\Re \zeta^2) + \tfrac{a}{2}(a-ib)}{{A}} x^2 \\
&\quad- \frac{W^4 + W^2(a\zeta^2 + (a+ib)\Re \zeta^2) + \frac{a}{2}(a+ib)}{{A}}  y^2 + 
\frac{2W^4}{{A}} xy \Big\}\\
&=  \sqrt{\tfrac{\pi}{{{A}}}}
 \exp\Big\{
 -\   
 \frac{\bar\alpha_{\hr }^{2} +\Re[W^{2}\zeta^{2} (a-ib)]+ \tfrac{1}{4}(a^{2}+b^{2}) }{ A} x^{2}
 \\
 &\quad  -\   
 \frac{\alpha_{\hr }^{2} +\Re[W^{2}\zeta^{2} (a-ib)]+ \tfrac{1}{4}(a^{2}+b^{2}) 
}{ {A} } y^{2}  -\frac{W^{4}}{{A}}   (x-y)^{2}\Big\}~. 
\end{split}\end{equation}
In particular, $K_\hr$ is normal ($K_\hr^*K_\hr = K_\hr K_\hr^*$) if and only 
if $\alpha_\hr^{2}$ 
is real (which happens if and only if $\alpha_\hr > 0$). More generally,
two operators of the form (\ref{eq:defHarm}) commute if and only if they share the same $\alpha_\hr$.

From (\ref{eq:K*K}), $K_\hr^*K_\hr$ is similar (conjugate) to the operator
\[\begin{split}
T_\hr(x, y) 
&=  e^{-i \frac{ \Im\alpha_{\hr }^{2} }{{A}}x^{2} }
(K_\hr^*K_\hr)(x, y)e^{+i \frac{ \Im\alpha_{\hr }^{2} }{{A}}y^{2} }\\
&= \sqrt{\tfrac{\pi}{{A}}} 
\exp\left\{
- \frac{W^4}{{A}} \, (x-y)^2
- 2a \, \left( 1 - \frac{a}{2A} \right) \,\frac{x^2 + y^2  }{2} \right\}
\end{split}\]
of the form (\ref{eq:defHarm}). This allows to compute the singular values
\[ s^{{\hr}}_0 \geq s^{{\hr}}_1 \geq \cdots \]
 of $K_\hr$: 
\[
\left(s^{{\hr}}_{j}\right)^{2} = \sqrt{\frac{\pi^{2}}{W^{4}  + 2a A + \sqrt{[2W^{4}+aA]aA} }}
\left( \frac{W^4}{ W^{4}  + 2a A + \sqrt{[2W^{4}+aA]aA} }   \right)^j~.
\]
 If $\alpha_\hr > 0$, we have $s^{{\hr}}_{j} = |\lambda_j^{{\hr}}|$
for any $j$. 
If {instead we} require $\alpha>0$,  then  $\alpha_{\hr}$ is  real up to an error term of order $O (W^{-2})$.
The corresponding operator is almost normal. More precisely, we have the following result.
\begin{lemma}\label{l:singharm}
If $K_\hr$ is an operator of the form (\ref{eq:defHarm}) with real (positive) $\zeta^2(a+ib)$, then
for any fixed $j$
\[ \frac{s^{{\hr}}_j}{|\lambda^\hr_j|} = 1 + O(W^{-2})~,\]
where the implicit constant may depend on $j$. Moreover,  for any $0<\epsilon<1$,
\[
\|K_\hr g_{\alpha }-\mu g_{\alpha }\| =|\mu | O (W^{-2+2\epsilon })  \qquad
\mbox{and}\qquad  \|{\tilde{g}}-g_{\alpha }\| \leq  O (W^{-2})~,
\]
where $\tilde{g}$ is the top normalized eigenfunction for $K_\hr^{*}K_\hr$,
and $g_{\alpha }$ and $\mu $ are  given by (\ref{eq:defgalpha}) and \eqref{eq:defAlpha}  respectively.
\end{lemma}
\begin{proof}
By the formul{\ae} for $\lambda_j^\hr$ and $s^{{\hr}}_{j}$ given above and using $\zeta^2 (a+ib)>0$
\[
\frac{|\lambda^{\hr}_{0}|^{2}}{\left(s^{{\hr}}_{0}\right)^{2}}= 
\sqrt{\frac{ W^{4}  + 2W^{3}\sqrt{a\Re \zeta^2   } + O (W^{2})  }{
W^{4}+ 2W^{3}\Re \zeta^2\sqrt{\zeta^2 (a+ib) } +O (W^{2})}}.
\]
Using the constraints $\Re \zeta^{2}>0$, $|\zeta^{2}|=1$ and $\zeta^2 (a+ib)>0$
we see that $\zeta^2 (a+ib) \Re \zeta^{2}=a$, whence 
$|\lambda^{\hr}_{0}|^{2}/\left(s^{{\hr}}_{0}\right)^2= (1+O (W^{-2}))$.
The same proof applies to the case $j>0$.
To prove the second part, we see that 
\[
(K_\hr g_{\alpha }) (x) = \mu g_{\alpha } (x)\, c (\alpha ) e^{-x^{2} d (\alpha )}~, 
\]
where
\[
c (\alpha )= \frac{1}{\sqrt{ 1+ \frac{a+ib}{2 (\alpha + W^{2} \zeta^{2} )}  }   }
= 1+ O (W^{-2}),
\quad d (\alpha )= \frac{ \frac{a+ib}{2}}{1+\frac{2[ \alpha +W^{2}\zeta^{2}] }{a+ib} } = O (W^{-2}).
\]
Then, using the exponential decay of $g_{\alpha }$,
\begin{multline*}
\|K_\hr g_{\alpha }-\mu g_{\alpha }\| \leq  |\mu | |c (\alpha )-1| + |\mu c (\alpha )|\, 
\left \|g_{\alpha }[ e^{-x^{2}d (\alpha )}-1 ]\mathbf{1}_{|x|\leq W^{\epsilon }} \right \|\\
 +  |\mu c (\alpha )|\, 
\left \|g_{\alpha }[ e^{-x^{2}d (\alpha )}-1 ]\mathbf{1}_{|x|> W^{\epsilon }} \right \| = 
|\mu |\,O (W^{-2+2\epsilon })~.
\end{multline*}
Finally, to prove the last inequality, we remark that
\[
\tilde{g} (x) = \left(\tfrac{2\alpha_{T} }{\pi } \right)^{1/4}  
e^{i \frac{ \Im\alpha_{\hr }^{2} }{A}x^{2} } e^{-\alpha_{T} x^{2}}~, 
\]
where $ \exp [-\alpha_{T} x^{2} ]$ is a top eigenfunction for $T_{h}$, and $\alpha_{T}$
is the real positive solution of  
\[
\alpha_{T}^{2}= \frac{W^{4}}{A} 2a \left(1- \frac{a}{2A} \right) + a^{2} 
\left(1- \frac{a}{2A} \right)^{2}= \alpha^{2} [1+O (W^{-2})]~.
\]
By assumption, $\Im \alpha_{\hr}^2 = O(1)$ and 
$\alpha_T =  \alpha (1 + O(1/W^{2})) =\alpha  + O(1/W)$, therefore
\[ \tilde{g}(x) = (1 + O(W^{-2})) e^{O(W^{-1})x^2} g_\alpha(x)~.\]
Hence 
\begin{align*}
\|\tilde{g}  -g_{\alpha }\|^{2} &\leq 
\int g_{\alpha } (x)^{2} 
\left | 1- e^{O(W^{-1}) x^{2}} \right|^{2} dx + O (W^{-4}) \\
& \leq O(W^{-2}) \int g_{\alpha }(x)^{2} x^{4}  e^{O(W^{-1}) x^{2}} dx + O (W^{-4})
= O(W^{-4}),\end{align*}
where in the last line we applied $\left | 1- e^{x} \right|\leq |x| e^{x}$.
\end{proof}

\section{Proof of the Main Proposition}\label{s:main}

Similarly to the semi-classical arguments in the self-adjoint case (see \cite[(5.6.1)]{H}), we
separate the contribution of the vicinity of the saddle point and the rest of the real line as
follows. Let $T(x, y)$ be a kernel, and suppose
$\chi_1^2 + \chi_2^2 = 1$ is a partition of unity. Then
\begin{equation}\label{eq:helff} T(x, y) = \sum_{j=1}^2 \chi_j(x) T(x, y) \chi_j(y) + 
\sum_{j=1}^2 R_j(x, y)~, \end{equation}
where
\[ R_j(x, y) = \frac{1}{2} (\chi_j(x) - \chi_j(y))^2 T(x, y)~. \]
In operator notation,
\[ T = \sum_{j=1}^2 \chi_j T \chi_j + \sum_{j=1}^2 R_j~. \]

\vspace{1mm}\noindent Another ingredient is Schur's bound (see \cite[Lemma 4.4.1]{H} for a proof)
\begin{equation}\label{eq:schur} 
\| T \| \leq \sqrt{\sup_x \int dy |T(x, y)|} \sqrt{\sup_y \int dx |T(x, y)|}~, \end{equation}
which, in the case when $|T(x, y)| = |T(y, x)|$,  assumes the form
\[ \| T \| \leq \sup_x \int dy \, |T(x, y)|~. \]

\vspace{1mm}\noindent
The difference from the usual setting stems from the fact that $K$ is not self-adjoint.
This is why we work with the self-adjoint operator $K^* K$, and our main effort will be invested in
decent bounds on the kernel.

\vspace{2mm}\noindent
The Main Proposition will follow from the next three lemmata, which are applied to estimate
the four blocks $A,B,C,D$ of (\ref{eq:matrix}).
We shall compare our operator $K$ with its harmonic approximation $\K$ introduced in (\ref{Ktilde}),
which is approximately normal due to assumption U2) of Section~\ref{sub:assump} and Lemma~\ref{l:singharm}.

\begin{lemma}\label{l:1}
Let $K$ be an operator given by (\ref{eq:defK}), so that $U$ satisfies the assumptions
U1) and U3). If $\alpha > 0$ is such that
\begin{equation}\label{eq:alphacond}
\left| \alpha^2 - W^2 \zeta^2 \frac{U''(0)}{2} \right| \leq CW^{3/2}~,
\end{equation}
then the asymptotics of the integral
\[ 
I(\alpha) = \iint dx dy \exp \left\{ -W^2 \zeta^2 (x-y)^2 - \frac{1}{2}U(x) - \alpha x^2
- \frac{1}{2} U(y) - \alpha y^2\right\}
\]
is given by
\[ 
I(\alpha) = (1 + O(W^{-3/2 + \epsilon})) \sqrt{\frac{\pi}{W^2\zeta^2+\alpha}}
\sqrt{\frac{\pi}{2\alpha}}~, 
\]
for any $\epsilon > 0$.
\end{lemma}
\begin{rmk}
Although the bound is valid for any $\alpha > 0$ satisfying (\ref{eq:alphacond}),
we shall only apply it to $\alpha = W \sqrt{\zeta^2 U''(0)/2}$ of (\ref{eq:defAlpha}).
\end{rmk}

\begin{lemma}\label{l:1.5}
In the setting of the Main Proposition,
\[ 
\| (K - \widetilde{K}) g_\alpha \| = O( W^{-3/2 + \epsilon} |\mu|)
\]
for any $\epsilon > 0$, where $\mu, g_{\alpha }$ and $ \widetilde{K}$ were introduced in
\eqref{eq:defAlpha} \eqref{eq:defgalpha} and \eqref{Ktilde}.
\end{lemma}

\begin{lemma}\label{l:2}
In the setting of the Main Proposition, let $u \in L_2$ be a function of unit norm.
Then there exists $\delta > 0$ so that
\[ \| Ku\| \leq |\mu| \, \left(1
+ O(W^{-1-\delta})\right)~. \]
{Moreover, if} $u \perp g_\alpha$, then
\[ \| Ku\| \leq |\mu| \, \left(1 - \sqrt{\frac{|U''(0)|}{2}} \frac{\Re \zeta^2}{W}
+ O(W^{-1-\delta})\right)~.\]
The same estimates hold for $\|K^* u\|$.
\end{lemma}

\begin{proof}[Proof of Main Proposition]
The estimate on $A$ follows from Lemma~\ref{l:1}:
\begin{equation}\label{eq:l1}
\mu A
= \langle K g_\alpha, g_\alpha \rangle = \sqrt{\frac{2\alpha}{\pi}} I(\alpha)
= (1+O(W^{-3/2+\epsilon})) \mu~.
\end{equation}
The estimate on $B$ and $C$ follows from Lemma~\ref{l:1.5}  {and \ref{l:singharm}}
\[\begin{split} \left\| \mu C \right\|
&= { \| P K g_\alpha\|= }\inf_{w \in \mathbb{C}} \| (K - w) g_\alpha\| \\
&\leq \| (K - \widetilde{K}) g_\alpha \| + \| \widetilde{K} g_\alpha - \mu  g_\alpha\|
= O(W^{-3/2 + \epsilon}|\mu|)~.\end{split}\]
{In a similar way  $\|\mu B\|=\|P K^{*}g_{\alpha }\|\leq \| (K^{*} - \bar\mu ) g_\alpha\|=
 \| (K- \mu ) g_\alpha\| $ since $g_{\alpha }$ is real. Therefore the arguments for $C$ apply.}
Finally, the bound on $\|D\|$ follows from the second statement of Lemma~\ref{l:2}, since
\[ \| PKP \| \leq \sup_{u \perp g_\alpha, \|u \| = 1} \| K u \|~.\]
\end{proof}

Now we turn to the proofs of the lemmata.

\begin{proof}[Proof of Lemma~\ref{l:1}]
Changing variables
\[ y \gets \frac{y+x}{\sqrt{2}}~, \quad x \gets  \frac{y-x}{\sqrt{2}}~, \]
we obtain:
\[ I(\alpha) = \iint dx dy \exp \left\{ - (2W^2 \zeta^2 + \alpha) x^2 - \alpha y^2 -
\frac{1}{2} U(\frac{y+x}{\sqrt{2}}) - \frac{1}{2}U(\frac{y-x}{\sqrt{2}})\right\}~.\]
The integration over the complement of the rectangle defined by the inequalities
$|x| \leq W^{-1+\epsilon/3}$, $|y| \leq W^{-1/2 + \epsilon/3}$ is exponentially suppressed
according to U3), where we took into account that $\alpha$ is of order $W$.
Inside the rectangle   we expand about $y/\sqrt{2}$ (since $x$ is typically
smaller in absolute value):
\[U(\frac{y\pm x}{\sqrt{2}}) = U(\frac{y}{\sqrt{2}}) \pm U'(\frac{y}{\sqrt{2}})  \frac{x}{\sqrt{2}}+ O(x^2)~,\]
whence  by U1)
\[ \frac{1}{2} \left( U(\frac{y+x}{\sqrt{2}}) + U(\frac{y-x}{\sqrt{2}}) \right)
	= U(\frac{y}{\sqrt{2}}) + O(x^2) = \frac{U''(0)}{4} y^2 + O(x^2 + |y|^3)~.\]
Therefore
\[ I(\alpha) = (1+O(W^{-3/2+\epsilon})) \sqrt{\frac{\pi}{2W^2\zeta^2 + \alpha}}
\sqrt{\frac{\pi}{\alpha + \frac{U''(0)}{4}}}~.\]
We have:
\[ (2W^2 \zeta^2 + \alpha)(\alpha + \frac{U''(0)}{4}) = 2W^2 \zeta^2 \alpha + \alpha^2 +
W^2 \zeta^2 \frac{U''(0)}{2} + O(W)~, \]
whereas
\[ 2\alpha (W^2 \zeta^2 + \alpha) =
2W^2 \zeta^2 \alpha + 2\alpha^2~.\]
Under the assumption (\ref{eq:alphacond}) on $\alpha$, the two expressions differ by $O(W^{3/2})$.
\end{proof}

\begin{proof}[Proof of Lemma~\ref{l:1.5}]
We start with the identity
\[ \| (K-\widetilde{K}) g_\alpha \|^2 = \int_{-\infty}^\infty \!\!dr \int_{-\infty}^\infty \!\!dx
\int_{-\infty}^\infty \!\!dy
\, (K^*-\widetilde{K}^*) (x, r)(K-\widetilde{K}) (r, y) g_\alpha(x) g_\alpha(y)~.\]
This time the integral over the complement of the polytope
\[ |x - r|, |r - y| \leq W^{-1 + \epsilon/3}~, \, |x|, |r| \leq W^{-1/2 + \epsilon/3}\]
is exponentially suppressed, whereas inside the polytope
\[ |(K^*-\widetilde{K}^*) (x, r)| \leq C W^{-3/2+\epsilon} |\widetilde{K}^* (x, r)|\]
and
\[ |(K-\widetilde{K}) (r, y)| \leq C W^{-3/2+\epsilon} |\widetilde{K} (r, y)|~;\]
the statement follows  from these inequalities.
\end{proof}

\vspace{1mm}\noindent
To prove  Lemma~\ref{l:2}, we need several estimates on the kernel of $K^*K$, which are collected
in the next lemma. We shall apply the first estimate when  $|x|,|y| \leq W^{-1/2+\delta}$, the second one
when either $\Re U \geq W^\eta$ or $|x-y| \geq W^{-1+\eta}$ (for a small $\eta > 0$ to be chosen later),
and the third one in the remaining range of parameters.

\begin{lemma}\label{l:K*K} The kernel $(K^*K) (x,y)$ satisfies the following estimates.
\begin{enumerate}
\item For $|x|, |y| \leq c_{0}$ (where $c_{0}>0$ 
may depend on $\zeta$ and on the width of the strip in which $U$ is analytic),
\[ (K^*K)(x, y) = \big[ 1 + O(|x|^3 + |y|^3 + W^{-3+\epsilon})\big] (\K^*\K)(x, y)~,\]
where  $\epsilon > 0$ is an arbitrary positive number,   and 
$\widetilde{K}$ was defined in \eqref{Ktilde}.
\item For any $x, y$,
\[ 
| (K^*K)(x, y)| \leq \sqrt{\frac{\pi}{2W^2 \Re \zeta^2}}
	\exp\left\{ - \frac{\Re U(x)+\Re U(y)}{2}  - \frac{W^2}{2} \Re \zeta^2 (x-y)^2 \right\}~.
\]
\item Let $\gamma $ be the parameter appearing in U4). If $|x-y| \leq W^{-1+ \eta}$,  
$\Re U(x) \leq W^{\eta}$ and $\Re U(y) \leq W^{\eta}$, 
where $\eta>0$  is sufficiently small,  then we have  
\begin{multline*} (K^*K)(x, y) = (1+O(W^{-1+5\eta\gamma})) \sqrt{\frac{\pi}{2W^2\Re\zeta^2}} \\
\times \exp\left\{
- \frac{\bar{U}(x) + U(y)}{2} - \Re U(\frac{x+y}{2}) - \frac{W^2}{2 \Re\zeta^2} (x-y)^2 \right\}~.
\end{multline*}
\end{enumerate}
\end{lemma}

We postpone the proof of Lemma~\ref{l:K*K} and start with
\begin{proof}[Proof  of Lemma~\ref{l:2}]
Let $\delta > 0$ be a small number. Construct a partition of unity $\chi_1^2 + \chi_2^2 = 1$. We choose
\[ \chi_1, \chi_2: \mathbb{R} \to \mathbb{R}_+ \]
such that
\begin{enumerate}
 \item $\chi_1$ is supported on $[ -W^{-1/2+\delta}, W^{-1/2+\delta}]$ and is identically equal to  one in
$[-\frac{1}{2}W^{-1/2+\delta}, \frac{1}{2}W^{-1/2+\delta}]$;
 \item $\chi_2$ is supported outside $[-\frac{1}{2}W^{-1/2+\delta}, \frac{1}{2}W^{-1/2+\delta}]$
and is identically equal to one outside $[ -W^{-1/2+\delta}, W^{-1/2+\delta}]$;
  \item the two functions are differentiable, and
$|\chi_1'|, |\chi_2'| \leq C W^{1/2 - \delta}$.
\end{enumerate}

Also denote $\ind_1 = \ind_{\chi_1 > 0}$, $\ind_2 = \ind_{\chi_2 > 0}$; then $\chi_1 \ind_1 = \chi_1$ and
$\chi_2 \ind_2 = \chi_2$.  According to the decomposition (\ref{eq:helff}),
\begin{equation}\label{eq:identity}
 \|K u \|^2 = \langle K^*Ku, u\rangle
= \sum_{j=1}^2 \langle \chi_j K^* K \chi_j u, u \rangle +
\sum_{j=1}^2 \langle R_j u, u \rangle~.
\end{equation}
Let $\s_0 \geq \s_1 \geq \cdots$ be the singular values of $\K$, and let $\g$
be the top eigenfunction of $\K^*\K$. From the properties of the harmonic oscillator
collected in Lemma~\ref{l:singharm},
\begin{equation}\label{smurel}
 \s_0 = |\mu| (1 + O(W^{-2}))~,
\quad \frac{\s_1}{\s_0} =  1 - \sqrt{\frac{|U''(0)|}{2}} \frac{\Re \zeta^2}{W} + O(W^{-2})~.
\end{equation}
We shall prove the following estimates:
\begin{align}\label{eq:t1}
\langle \chi_1 K^* K \chi_1 u, u \rangle
&\leq \s_1^2 \|\chi_1 u\|^2 + C \s_0^2 W^{-3/2+3\delta} \|u\|^2 \quad (u \perp g_\alpha)~, \\
\label{eq:t1*}
\langle \chi_1 K^* K \chi_1 u, u \rangle &\leq \s_0^2 (1 + CW^{-3/2+3\delta}) \|\chi_1 u\|^2~, \\
\label{eq:t2}
|\langle \chi_2 K^* K \chi_2 u, u \rangle | &\leq \s_0^2 (1 - \frac{1}{C} W^{-1+2\delta}) \|\chi_2 u\|^2~, \\
\label{eq:t3}
|\langle  R_j u, u \rangle | &\leq CW^{-1-2\delta} \s_0^2\|u\|^2~.
\end{align}
Once these bounds are established, the proof of the lemma is concluded as follows.
For $u \perp g_\alpha$, we use (\ref{eq:t1}), (\ref{eq:t2}), (\ref{eq:t3}) to estimate the addends in
(\ref{eq:identity}); then from the inequality
\[ 
\s_0^2 (1 - \frac{1}{C W^{1-2\delta }}) \leq 
\s_0^2 (1 - 2\sqrt{\tfrac{|U''(0)|}{2}} \tfrac{\Re \zeta^2}{W} + O(W^{-2}))= \s_1^2 
\]
and the identity
\[ 
\|\chi_1 u\|^2 + \|\chi_2 u\|^2 = \|u\|^2 
\]
we obtain:
\[ 
|\langle K^* K u, u \rangle | \leq (1+O(W^{-1-2\delta}))\s_1^2 \|u\|^2~.
\]
For arbitrary $u$ we apply  (\ref{eq:t1*}) in place of (\ref{eq:t1}), and obtain:
\[ 
|\langle K^* K u, u \rangle | \leq (1+O(W^{-1-2\delta}))\s_0^2 \|u\|^2~.
\]
Inserting now  \eqref{smurel}, the proof is concluded. 

\vspace{2mm}\noindent
{\bf Proof of (\ref{eq:t1})} starts from
\begin{equation}\label{eq:tmp}\begin{split}
&\langle \chi_1 K^* K \chi_1 u, u \rangle \\
&\qquad= \langle \K^* \K \chi_1 u, \chi_1 u \rangle
 + \langle (\ind_1 K^* K \ind_1 - \ind_1 \K^* \K \ind_1) \chi_1 u, \chi_1 u \rangle~. \end{split}\end{equation}
From the decomposition
\[ \chi_1 u = \langle \chi_1 u, \, \g \rangle \g + \big\{ \chi_1 u - \langle \chi_1 u, \, \g \rangle \g \big\} \]
we obtain
\[\langle \K^* \K \chi_1 u, \chi_1 u \rangle  \leq \s_0^2 |\langle \chi_1 u, \g \rangle|^2
	+ \s_1^2 \|\chi_1 u \|^2~.\]
If $u \perp g_\alpha$,
\begin{align*} 
|\langle \chi_1 u, \g \rangle |& \leq \|\chi_1u\| \|\g - g_\alpha\| +
|\langle u - \chi_1 u,  g_\alpha\rangle| \\
&{\leq  \|u\| \|\g - g_\alpha\|+  \|u\| \|(1-\chi_{1}) g_\alpha\|}
 \leq C_{1}W^{-2} \|u\|~, 
\end{align*} 
where in the first term of the sum we applied 
Lemma~\ref{l:singharm}.
{Finally}
\begin{equation}
\label{eq:nearsaddle}
\|\ind_1 K^* K \ind_1 - \ind_1 \K^* \K \ind_1 \|
\leq C_{2} W^{-3/2 + 3\delta} \s_0^2
\end{equation}
according to item 1.\ of Lemma~\ref{l:K*K} and Schur's bound (\ref{eq:schur}).
Hence
\begin{align}\label{eq:t1a}
 \langle \chi_1 K^* K \chi_1 u, u \rangle 
&\leq C_{1}^{2} W^{-4} \|u\|^{2} + \s_1^2 \|\chi_1 u \|^2 + C_{2} \s_0^2 
W^{-3/2+3\delta} \|u\|^2 \\
& \leq \s_1^2 \|\chi_1 u\|^2 + C \s_0^2 W^{-3/2+3\delta} \|u\|^2~.
\end{align}
 
\vspace{1mm}\noindent {\bf Proof of (\ref{eq:t1*})} also starts from (\ref{eq:tmp}). From
\[ \langle \chi_1 \K^* \K \chi_1u, u \rangle \leq \s_0^2 \|\chi_1 u \|^2\]
and (\ref{eq:nearsaddle}),
we obtain the bound
\begin{equation}\label{eq:t1*a}
\langle \chi_1 K^* K \chi_1 u, u \rangle \leq \s_0^2 (1 + CW^{-3/2+3\delta}) \|\chi_1 u\|^2~.
\end{equation}

\vspace{1mm}\noindent {\bf Proof of (\ref{eq:t2})}  
We plug  into Schur's bound (\ref{eq:schur}) the estimates on the kernel
of $(K^* K) (x,y)$ obtained in Lemma~\ref{l:K*K},  as follows. Set 
$\eta = \delta/ (5\gamma)$. We use the estimate given in  item 2 when  either
$|x-y|>W^{-1+\eta }$, or $\Re U (x)>W^{\eta }$, or $\Re U (y)>W^{\eta }$. In the complementary region 
\[ |x-y|\leq W^{-1+\eta} \quad \wedge \quad \Re U (x)\leq W^{\eta } \quad \wedge  \quad \Re U (y)\leq W^{\eta }\]
we use  the estimate given in  item 3.  Then 
\begin{align*}
\|\ind_2 K^* K \ind_2 \| &\leq \ 
\tfrac{\pi }{W^{2}}e^{-\frac{W^{-1+2\delta }}{C}} (1+O (W^{-1+ 5 \eta \gamma})) 
+ e^{-C_{1} W^{\eta }}\\
&\leq   
\s_0^2 (1 - \frac{1}{C} W^{-1+2\delta})~, 
\end{align*}
where in the first term we used $\Re U (x)\geq \frac{|x|^{2}}{C}\geq\frac{W^{-1+2\delta }}{C} $ (from U3) and
the definition of $\chi_{2}$) and 
the relations between $|\mu |= \frac{\sqrt{\pi}}{W} (1+O (W^{-1}))$ and $\tilde{s}_{0}$ given 
in \eqref{smurel}.
The second term comes from the estimate in item 2 of Lemma~\ref{l:K*K}. Hence
\begin{equation}\label{eq:t2a}
|\langle \chi_2 K^* K \chi_2 u, u \rangle | \leq \s_0^2 (1 - \frac{1}{C} W^{-1+2\delta}) \|u\|^2~.
\end{equation}

\vspace{1mm}\noindent {\bf Proof of (\ref{eq:t3})} is similar: Schur's bound, 
{item 2 of} Lemma~\ref{l:K*K} and the
bounds
\[ |\chi_j(x) - \chi_j(y)| \leq C W^{1/2 - \delta} |x-y|\]
are used to show that
\[ \|R_j\| \leq CW^{-1-2\delta}\s_0^2~.\]
\end{proof}

\begin{proof}[Proof of Lemma~\ref{l:K*K}]
First,
\[ (K^*K)(x,y) = E(x, y) I(x, y)~, \]
where
\[ E(x, y) = \exp \left\{ - \frac{\bar{U}(x)+ U(y)}{2} - W^2 \bar{\zeta}^2 x^2
- W^2 \zeta^2 y^2\right\}~,\]
and
\[ I(x,y) = \int dr \exp \left\{ - 2W^2 \left[ \Re \zeta^2 r^2 - (\bar\zeta^2 x + \zeta^2 y) r\right]
- \Re U(r) \right\}~.\]
On the real line, $\Re U(z)$ coincides with the analytic function 
\begin{equation}\label{eq:uschw} 
U_\text{Schw}(z) = (U(z) + \overline{U(\bar{z})})/2~, 
\end{equation}
therefore we replace   $\Re U$ with
$U_\text{Schw}$.

To prove the first item of the lemma,  let $|x|,|y|<c_0$ for a small constant $c_0$, and 
let $\E$ and $\I$ be expressions analogous to $E$ and $I$ which correspond to  $\K$;
\[ \I(x, y) = \sqrt{\frac{\pi}{2W^2 \Re \zeta^2 + \frac{\Re U''(0)}{2}} }
\exp \left\{ \frac{W^4 (\bar\zeta^2 x + \zeta^2 y)^2}{2W^2\Re \zeta^2
+ \frac{\Re U''(0)}{2} } \right\}~.\]
Let   $\xi \geq 2|x| + 2|y|$ be a small number to be fixed later on in
(\ref{eq:defxi}), and set 
\[ r_{0}=\frac{W^{2} (\bar\zeta^2 x + \zeta^2 y)}{2W^2\Re \zeta^2
+ \frac{\Re U''(0)}{2} }~.\]
Then $|r_{0}|\leq \xi/4$. Deform the contour of integration to
\[ (-\infty, \Re r_0-\xi) \cup (\Re r_0-\xi, r_0 - \xi) \cup (r_0 - \xi, r_0 + \xi) 
\cup (r_0 + \xi, \Re r_0 + \xi) \cup (\Re r_0 + \xi, \infty)~.\]
Let $I = I_1 + I_2$, where $I_1$ is the integral over  $(r_0 - \xi, r_0 + \xi)$, and
$I_2$ is the integral over the remaining part of the contour. Let $\I = \I_1 + \I_2$ be the analogous
decomposition of $\I$. (Observe that, for sufficiently small $c_0$, the deformed contour is
within the domain of analyticity of $U$.)
Then
\[ |I_2|,|\I_2| \leq \exp \left\{ - C^{-1} W^2 \xi^2 \right\}~. \]
To estimate the difference between the dominant parts $I_1,\I_1$, we write
\[
 I_1 - \I_1 = \exp {\left\{ \frac{W^4 (\bar\zeta^2 x + \zeta^2 y)^2}{2W^2\Re \zeta^2
+ \frac{\Re U''(0)}{2} }\right\}} \int\limits_{r_0 - \xi}^{r_0 + \xi} dr \, 
e^{- \left[ 2W^2 \Re \zeta^2+ \frac{\Re U''(0)}{2}
\right]
 (r-r_{0})^{2}  }\left[ e^{R (r)  }-1 \right]
\]
where 
\[
R (r)= \frac{U''_{\text{Schw}}(0)}{2} r^{2}- U_{\text{Schw}}(r).
\]
We obtain:
\[ |I_1 - \I_1| = O(( |r_{0}| +\xi)^3) |\I|~, \]
To conclude the proof of the first item, set
\begin{equation}\label{eq:defxi}
\xi = 2(|x|+|y|) + W^{-1+\epsilon/3}~,
\end{equation}
and observe that
\[ E(x, y) = (1+O(|x|^3+|y|^3)) \E(x, y)~. \]

\vspace{2mm}\noindent
To prove the second item, we insert  absolute values:
\[ |E(x, y)| \leq \exp \left\{ - \frac{\Re U(x) + \Re U(y)}{2} - W^2 \Re \zeta^2 (x^2 + y^2)\right\}~,\]
and
\[\begin{split}
|I(x, y)| &\leq \exp \left\{ \frac{W^2}{2} \Re \zeta^2 (x+y)^2 \right\}\\
&\qquad\times \int dr \exp\left\{  - 2W^2 \Re \zeta^2 (r - \frac{x+y}{2})^2 - \Re U(r) \right\}\\
&\leq \sqrt{\frac{\pi}{2W^2\Re\zeta^2}} \exp \left\{ \frac{W^2}{2} \Re \zeta^2 (x+y)^2 \right\}~.
\end{split} \]

\vspace{2mm}\noindent
To prove the third item, let 
 us rewrite $I(x,y)$ as
\[ \exp \left\{ \frac{W^2}{2} \frac{(\bar\zeta^2 x + \zeta^2 y)^2}{\Re \zeta^2}\right\}
\int dr \exp\left\{ - 2W^2 \Re \zeta^2 (r-r_0)^2 - \Re U(r)\right\}~,\]
where
\[ r_0 = \frac{y+x}{2} + \frac{\Im \zeta^2}{\Re \zeta^2} \frac{y-x}{2} \, i~. \]
Then performing a contour deformation similar to the one in  the proof of the first item we have
\[
\int dr \exp\left\{ - 2W^2 \Re \zeta^2 (r-r_0)^2 - \Re U(r)\right\} = I_{1}+ I_{2}~,
\]
where
\begin{align}\label{eq:i1i2}
I_{1}&=  e^{- U_{\text{Schw}}(r_{0})} 
\int_{-\xi}^{\xi} dr e^{ - 2W^2 \Re \zeta^2 r^2  } 
e^{ U_{\text{Schw}}(r_{0})-U_{\text{Schw}}(r+r_{0}) } 
 \\
|I_{2}|&\leq  e^{-W^{2} (\xi^{2}- |\Im r_{0}|^{2})/C}\leq e^{-W^{2} \xi^{2}/C'} ~,
\end{align}
if we choose $\xi>|\Im r_{0}|/4$.  The imaginary part $ |\Im r_{0}|$ may be as 
large as $W^{-1+\eta} $, therefore
we have to take $\xi= W^{-1+\eta +\epsilon }$ for some $\epsilon > 0$.
 We later set $\epsilon = 2\eta\gamma$.

Let us show that for any $\eta_1 \in (\eta, 2\eta)$  the following estimate holds
for  $r$ in a complex neighbourhood of $x$:
\begin{equation}\label{eq:derbd}
|r - x| \leq 2 \xi= 2 W^{-1+\eta +\epsilon }\Longrightarrow \quad |U'(r)| 
< 2 C W^{\eta_1 \gamma}~.
\end{equation}
Indeed,  by the inequalities $\Re[U (x)]\leq W^{\eta }$ and U4), $U' (x)$ satisfies  
$ |U'(r)| < 2 C W^{\eta \gamma}$, and the smoothness of 
$U$ guarantees there exists some constant $c_{x}>0$ such that 
(\ref{eq:derbd}) holds inside the ball of radius $c_{x}$ centred at $x$. Let 
  $r_{1}  \in \mathbb{C}$ be a point such that  $|r_{1}-x|> c_{x}$, (\ref{eq:derbd}) holds for all 
$ |r - x|< |r_{1} - x|$ and fails at $r_{1}$.  Then by U4)
\begin{equation}\label{eq:fromU4}
\Re[U(r_{1})]  > W^{\eta_{1} }~.
\end{equation}
Performing a Taylor expansion with first order integral remainder we have
\[
 \frac{U(r_{1})-U(x)}{r_{1}-x}= \int_{0}^{1}  U'(x+t (r_{1}-x))dt.
\]
Inserting absolute values,
\[
 \frac{|U(r_{1})-U(x)|}{|r_{1}-x|}\leq  \int_{0}^{1}  |U'(x+t (r_{1}-x)|dt
\leq  2 C W^{\eta_1 \gamma}
\]
since $|x+t (r_{1}-x)|< |r_{1} - x|$ for all $0\leq t<1$. 
From  (\ref{eq:fromU4}) and the assumptions  $ \Re[U(x)] \leq  W^{\eta}\ll W^{\eta_{1}} $
and $|r_{1}-x|\leq 2W^{-1+\eta+\epsilon  }$ we get
\[
\frac14\, W^{\eta_{1}+1-\eta -\epsilon } \leq
\frac{   |\Re[U(r_{1})]- \Re[U(x)]| }{|r_{1}-x|} \leq 
\frac{|U(r_{1})-U(x)|}{|r_{1}-x|}\leq   2 C W^{\eta_1 \gamma}
\]
hence 
\begin{equation}\label{eq:contr}
 W^{1-\eta -\epsilon } \leq 8C  W^{\eta_1 (\gamma-1)}
\end{equation}
as long as $\eta_{1} (\gamma -1 )< 1-\eta -\epsilon$.  When $\eta,\epsilon > 0$ are 
sufficiently small we have  $\eta_{1} (\gamma -1 )< 1-\eta -\epsilon$ 
for all $\eta_1 \in (\eta, 2\eta)$, in contradiction with (\ref{eq:contr}).
Thus (\ref{eq:derbd}) is established. 

Applying  the definition (\ref{eq:uschw}) of $U_\text{Schw}$  we  have 
$U'_\text{Schw}(r)=(U'(r)+\overline{U'(\bar r)})/2$, and 
from (\ref{eq:derbd}) 
\[
 |r - x| \leq 2 \xi= 2 W^{-1+\eta +\epsilon }\Rightarrow \quad |U'_\text{Schw}(r)| 
< 2 C W^{\eta_1 \gamma}~,
\]
where $ |\bar{r} - x|= |r- x|$ since $x\in \mathbb{R}$. 
Now set $\epsilon =2\eta \gamma $. Then for any $\eta <\eta_{1}<2\eta $
we have $\epsilon >\eta_{1}\gamma $ and 
\begin{align*}
 |U_\text{Schw}(r_{0})-U_\text{Schw}(r)| &\leq |r-r_{0} |\int_{0}^{1}  
|U'_\text{Schw}(r+t (r_{0}-r))|dt\\
&  \leq O (W^{-1+\eta +\epsilon+\gamma \eta_{1}})\ = O (W^{-1+5\gamma\eta})~,
\end{align*}
hence  $I_1$ of (\ref{eq:i1i2}) satisfies
\[
I_{1}= e^{- \Re U( \frac{x+y}{2})} \sqrt{\frac{\pi}{2W^2 \Re \zeta^2}} \left[ 1+O (W^{-1+5\gamma\eta })  \right]~.
\]
This concludes the proof of the third item  of Lemma~\ref{l:K*K}.

\end{proof}

\section{Application to a complex statistical mechanics model}\label{s:applic}

In this section, we apply the results of the paper
to a toy model. The model is tailored so that 
the conditions U1)--U4) of the Main Proposition
will be satisfied after a rotation of the integration contour.
The choice of the potential is partly inspired by supersymmetric
models appearing in the study of random operators, but in our case 
the potential has only one minimum, instead of several minima as in the
original models.

Let $V(x) = a \log (1 + bx^2)$, where $a > 0$ and $\Re b > 0$. 
We are interested in the statistical mechanics model corresponding to the action
\[ W^2 \sum_j (\phi_j - \phi_{j+1})^2 + \sum_j V(\phi_j)~;\]
for simplicity of notation, we set the inverse temperature to one.
Without going into the details of the construction of infinite-volume measures (which 
is impeded by several obstacles, see e.g.\ Remark~\ref{rmk:bdry} below) , let us
define the ``mean'' of a local observable $F:\mathbb{R}^{n+1} \to \mathbb{C}$
as follows:
\begin{multline}\label{eq:mean}\langle F(\phi_0, \cdots, \phi_n) \rangle
= \\ \lim_{M,N \to \infty} \frac
{\int \prod_{j=-M}^N d\phi_j e^{-\sum_{j=-M}^N V(\phi_j) - W^2 \sum_{j=-M}^{N-1}
(\phi_j - \phi_{j+1})^2} F(\phi_0, \cdots, \phi_n)}
{\int \prod_{j=-M}^N d\phi_j e^{-\sum_{j=-M}^N V(\phi_j) - W^2 \sum_{j=-M}^{N-1}
(\phi_j - \phi_{j+1})^2}}~.
\end{multline}
We are interested in the long-distance correlations, e.g.\
\[ 
\langle (F(\phi_0) - \langle F(\phi_0) \rangle) (G(\phi_n) - \langle G(\phi_n) \rangle) 
\rangle~, 
\]
with  $F,G:\mathbb{R} \to \mathbb{C}$. 
Define $\zeta \in \mathbb{C}$ by
\begin{equation}\label{eq:defzeta}
|\zeta|=1~, \quad |\arg \zeta| < \pi/4, \quad \zeta^4 V''(0) \, 
\big(= 2 \zeta^4 \, ab\, \big) \,> 0~; 
\end{equation}
let $\Sigma = \mathrm{conv} \left(\mathbb{R} \cup \mathbb{R}\zeta \right)$, and set
\begin{equation}\label{eq:defU}
U(x) = V(\zeta x)~. 
\end{equation}

\vspace{2mm}\noindent
Our transfer operator method can be applied to study observables 
$F:\mathbb{R}^{n+1} \to \mathbb{C}$
which have an analytic extension to $\Sigma^{n+1}$ and do not grow too fast in this sector. For
simplicity, let us focus on $n=0$, i.e.\ on observables
depending only on one variable: assume that
\begin{enumerate}[F1)]
\item $F: \Sigma \to \mathbb{C}$ is analytic;
\item $|F(z)| \leq C (1 + |z|)^{2a  - 1 - \epsilon}$ for some $C > 0$ and $\epsilon > 0$, 
and all $z \in \Sigma$.
\end{enumerate}

\begin{thm} Suppose that $V(x) = a \log(1 + bx^2)$ for some $a > 0$ and $\Re b > 0$, and that 
$F, G$ are observables which satisfy F1)--F2).
Then
\begin{multline*}
|\langle (F(\phi_0) - \langle F(\phi_0) \rangle) (G(\phi_n) - \langle G(\phi_n) \rangle) \rangle|
	\\\leq C_F C_G \left(1 - \sqrt{\frac{|V''(0)|}{2}} \frac{\Re \zeta^2}{W} + O(W^{-1-\delta})\right)^n~,
\end{multline*}
where $\zeta$ is defined in (\ref{eq:defzeta}).
\end{thm}

\begin{proof} Let $F, G$ be observables satisfying
F1)--F2). Note that the functions inside both integrals of  (\ref{eq:mean})
satisfy $f (\phi_{-M},\dotsc ,\phi_{N})\sim e^{-W^{2}\phi_{j}^{2}}$ as $|\phi_{j}|\to\infty$
while keeping the other variables fixed. Then  using Cauchy's theorem we can rotate the 
contour of integration $\phi_{j}\to \phi_{j}\zeta$, as long as $\Re \zeta^2 > 0$. Repeating
the argument for each variable $\phi_{j}$ we obtain
\begin{equation}\label{eq:mean'}
\begin{split}
&\frac
{\int \prod_{j=-M}^N d\phi_j e^{-\sum_{j=-M}^N V(\phi_j) - W^2 \sum_{j=-M}^{N-1}
(\phi_j - \phi_{j+1})^2} F(\phi_0) G(\phi_n)}
{\int \prod_{j=-M}^N d\phi_j e^{-\sum_{j=-M}^N V(\phi_j) - W^2 \sum_{j=-M}^{N-1}
(\phi_j - \phi_{j+1})^2}} \\
&\quad=
\frac
{\int \prod_{j=-M}^N d\phi_j e^{-\sum_{j=-M}^N U(\phi_j) - W^2 \zeta^2 \sum_{j=-M}^{N-1}
(\phi_j - \phi_{j+1})^2} F(\zeta \phi_0) G(\zeta \phi_n)}
{\int \prod_{j=-M}^N d\phi_j e^{-\sum_{j=-M}^N U(\phi_j) - W^2 \zeta^2 \sum_{j=-M}^{N-1}
(\phi_j - \phi_{j+1})^2}}~.
\end{split}
\end{equation}

\begin{rmk}\label{rmk:bdry}
The argument presented above for free boundary conditions applies equally well to periodic
boundary conditions. For more general boundary 
conditions the potential at the boundary is modified, therefore additional requirements on $U$ 
and the observables
probably have to be imposed.
\end{rmk}

Now we set up an integral operator $K$, the kernel of which is given by (\ref{eq:defK});
the Main Proposition and Corollary~\ref{cor} are applicable, so the largest  eigenvalue
of $K$ (in absolute value)
$\lambda_0$ is given by (\ref{eq:lam0}). Let $u_{0}$ be a corresponding eigenfunction
satisfying the normalisation conditions \eqref{normcond}. 
Set $B(x) = \exp(-U(x)/2)$. Then the denominator of (\ref{eq:mean'}) is equal to
\[ 
\int K^{M+N}(x_1, x_2) B(x_1) B(x_2) dx_1 dx_2~. 
\]
From  Corollary~\ref{cor} we have 
\[
\langle u_0,\bar{u}_0 \rangle=\langle u_0, g_{\alpha } \rangle + 
\langle u_0,\bar{u}_0-g_{\alpha } \rangle = 1 + O (W^{-\delta })\neq 0. 
\]
Decomposing $B = \frac{\langle B,\bar{u}_0 \rangle}{\langle u_0,\bar{u}_0 \rangle }u_0 + B_1$ 
(this is well defined since $\langle u_0,\bar{u}_0 \rangle=1 + O (W^{-\delta })  \neq 0$)
and setting
$\hat{K} = \lambda_0^{-1} K$, we obtain from Corollary~\ref{cor}:
\[ 
\lim_{M,N \to \infty} \int \hat{K}^{M+N}(x_1, x_2) B(x_1) B(x_2) dx_1 dx_2
 = 
\frac{\langle B,\bar{u}_0 \rangle^{2}}{\langle u_0,\bar{u}_0 \rangle^{2} }~.
\]
Similarly, for $F: \mathbb{R} \to \mathbb{C}$ satisfying F1)--F2),
\[ 
\int \prod_{j=1}^3 dx_j B(x_1) \hat{K}^M(x_1, x_2) F(\zeta x_2)
\hat{K}^N(x_2, x_3) B(x_3)
\to 
\frac{\langle B,\bar{u}_0 \rangle^{2}}{\langle u_0,\bar{u}_0 \rangle^{2} }
\langle F_\zeta u_0, \bar{u}_0 \rangle~,
\]
where $F_\zeta(x) = F(\zeta x)$. Hence we obtain:
\begin{equation}
\langle F(\phi_0) \rangle = \langle F_\zeta u_0, \bar{u}_0 \rangle~.
\end{equation}
Similarly,
\begin{equation}
\langle F(\phi_0) G(\phi_n) \rangle = 
\langle \hat{K}^n F_\zeta u_0, \bar{G}_\zeta \bar{u}_0 \rangle~.
\end{equation}
If $\langle F(\phi_0) \rangle = 0$ (i.e.\ $ F_\zeta u_0\in \bar{u}_0^{\perp} $), 
Corollary~\ref{cor} yields:
\[ 
\| \hat{K}^n F_\zeta u_0 \| \leq \left(1 - \sqrt{\frac{|V''(0)|}{2}} \frac{\Re \zeta^2}{W} 
+ O(W^{-1-\delta})\right)^n \|F_\zeta u_0\|~,
\]
thus for any $F, G$ satisfying F1)--F2)
\begin{multline*}
|\langle (F(\phi_0) - \langle F(\phi_0) \rangle) (G(\phi_n) - \langle G(\phi_n) \rangle) \rangle|
	\\
\leq C_F C_G \left(1 - \sqrt{\frac{|V''(0)|}{2}} \frac{\Re \zeta^2}{W} 
+ O(W^{-1-\delta})\right)^n~.
\end{multline*}
\end{proof}

\paragraph{Acknowledgments} We are  grateful to Tom Spencer for encouraging 
us to pursue the transfer matrix approach, and for numerous suggestions.
We  thank  Sergey Denisov for  discussions on
non self-adjoint operators, and  Tanya Shcherbina for critical comments.

Margherita Disertori thanks the Institute for Advanced Study for 
 hospitality while some of this work was in progress, and 
acknowledges partial support from ERC Starting Grant CoMBos. 
Sasha Sodin acknowledges  partial support by NSF under grant PHY-1305472.

\end{document}